\documentclass[10pt]{article}

\usepackage{amsmath}
\usepackage{amsthm}
\usepackage{amsfonts}
\usepackage{enumerate}

\newtheorem{remark}{Remark}
\newtheorem{fact}{Fact}
\newtheorem{theorem}{Theorem}

\begin{document}

\title{Optimizing On-Line Advertising}

\author{Fabrizio Caruso\\ Neodata Group, Catania, Italy \\ e.mail: \texttt{fabrizio.caruso@neodatagroup.com} \\
\and    Giovanni Giuffrida\\ Dept. of Social Sciences, University of Catania, Italy \\ e.mail: \texttt{ggiuffrida@dmi.unict.it}
       }

\maketitle

\begin{abstract}
We want to find the optimal 
strategy for displaying advertisements
e.g.\ banners, videos, 
in given locations at given times
under some realistic dynamic constraints. 
Our primary goal is to maximize the expected revenue
in a given period of time, i.e.\ the total profit
produced by the impressions, which depends on
profit-generating events such as the impressions themselves,
the ensuing clicks and registrations.
Moreover we must take into consideration
the possibility that the constraints could 
change in time in a way that cannot always be
foreseen.


\end{abstract}

\section{Introduction}
We want to find the optimal strategy for displaying (delivering)
advertisements (``creatives'') in order to achieve different goals
(maximum total profit, maximum average profit per advertisement, maximum
visibility of the campaign) under some realistic constraints.

More specifically we need to find the  optimal number of ``impressions'' 
(display of advertisements at a given time and at a given ``location'') 
under some realistic dynamic constraints that both limit the
possibility of certain creatives 
and limit the number of impressions 
in certain locations and/or moments in time.
A location can represent a place where an advertisement can be displayed.
It could also include information on the user or user's category
(it could be a combination of a location and user category,
location and a user, location and a set of keywords inserted into the
web-page, etc\dots).
Hence the model we are considering can be used to target users
by their categories.
Similar optimization problems have already been treated in the
scientific literature (\cite{unintru}, \cite{abenakamura}, \cite{nakamura}).

We are given certain 
``creatives'' (advertisements to be displayed, e.g.\ banners, videos, etc.),
``campaigns'' (sets of related creatives),
certain ``locations'' and a period of time (set of ``time frames''). 
At a given moment in time we have an expected profit
for each creative of a given campaign in a 
given time frame and location.
Usually the profit generated by 1000 impressions, called ``eCPM'' 
(effective cost per mille) is considered
because the profit of one impression is very small.
For the purpose of this paper 
we will only consider the profit of a single impression.

The profit of the web-page's owner depends on the 
profit-generating events
that have been agreed upon by the advertiser and the
web-page's owner.
These events can be the impression itself,
a click on the advertisement or a registration of any sort
(e.g.\ registration into the advertised site, 
purchase of the advertised item, etc\dots), or any
combinations of these events.

We denote the expected profit of a single ``impression''
as the ``impression profit''.
The impression profit is therefore the sum of the
profits obtained by all the profit-generating events
such as impressions themselves, 
their ensuing clicks and registrations of all types (``steps'').
In such a way we can avoid keeping track of click-through
rates and different registration rates.
This choice is a compromise between performance and
generality,  
since it makes our model less precise and slightly less general:
we are not considering campaigns with separate budgets 
for different events;
we cannot estimate the expected profit of an impression 
as precisely as when different rates for different events are
considered.

The number of impressions (``supply'') on a given location
at a given time is limited by the traffic of the corresponding
webpage. It also depends on time
in a way that can be only partially predicted.
Moreover the maximum profit for a given campaign (``demand'') 
could be limited by a predefined budget.

Our primary goal is to 
maximize our expected revenue which is given by 
the expected total price paid.
A secondary goal is to maximize the profit of a single impression,
i.e.\ obtain the maximum revenue with the 
minimum effort (minimum number of
impressions). 
Therefore we wish to maximize 
a weighted sum of all all expected profits obtained 
in all locations in the period of time under consideration.
An additional goal which is considered in
the constraints is to maximize the visibility of the campaigns.

Taking into account supply and demand constraints makes our model
a special instance of a ``transportation problem''
for which very efficient solutions exist (see \cite{dantzig} and \cite{transshipment}).

The complexity of the model brings up the additional problem of 
deciding between simplifying the model 
and considering smaller problems, i.e.\ optimizing more locally.

In order to apply our optimization we need to make a
projection of the future supply and 
a projection of the impression profits
onto our period.
Impressions are only possible on the points allowed by the scheduling
of the campaigns.

The projection of the impression profits 
should also try to ``guess'' how the profit of an impression
changes in time.

The projection algorithms should take into account different
periodicities (daily, weekly, etc\dots).
More precise projections can be achieved through machine learning
techniques in that the weights of different periodicities
are computed during a training phase, possibly with different
strategies for different nodes or other features (time of the
day, day of the week, etc\dots).

Moreover we cannot
assume the immutability of the constraints of the problem 
in the period of time under consideration.
For this reason we have to decide how globally or locally we want to
optimize the problem and we have to continuously readjust to 
new constraints, new expected impression profits and
new expected supply.

\section{Notation}
We denote by $C_i$ the $i$-th campaign (set of creatives) and
by $B_{i,j}$ its $j$-th creative, 
by $L_l$ the $l$-th location,
by $T_k$ the $k$-th time frame.
We denote by $x_{i,j,k,l}$ the (``impression value'') number of impressions  
of $B_{i,j}$ at time frame $T_k$ and at location $L_l$.

\noindent We denote by $p_{i,j,k,l}$ 
the ``impression profit'', the profit of one single impression of creative 
$B_{i,j,k,l}$ at position $L_l$ and a time $T_k$.

\noindent For the sake of simplicity we will be omitting indices 
in our expression whenever our
objects do not depend on them.

\section{Configurations}

We can consider our problem as the problem
of finding the optimal impression values for the entries
in a tridimensional matrix, i.e.\  
for all points in a tridimensional discrete finite space
given by a grid defined by
couples (campaign, creative), 
time and location.

We refer to a single point in this tridimensional
discrete finite space
as an ``impression-event'' (or simply an ``impression''
when this is clear from the context). 
We call any choice for the values
of all the impression-events as a ``configuration''.
An impression is in fact
characterized by a couple (campaign, creative),
location and time. Our goal is to choose
the optimal delivery of each possible
impression, i.e.\ an optimal configuration. 
We will simply refer to the the number of impressions
of an impression-event as
the ``impression value'' or its ``value''.

Moreover some of our constraints restrict
the possible points in such a grid.
We can see these restriction as an 
additional trivial constraints
of the form $x_{i,j,k,l}=0$
on such points. 
We refer to the points
that do not contradict any constraints as
``possible points'', points that are not
forced to have a zero value by the constraints. 
A subset of the possible points is the set
of points that are allowed by the schedule
of the campaigns. We refer to such points as to
``admissible points''.

Each admissible point in this space describes a dimension
of our optimization problem.
The worst case in our problem is produced when
all points inside the cube of size given
by the number of couples (campaign, creative),
the number of time frames, the number of locations, are admissible.
Our problem  lives in a space whose number of dimensions 
is given by the number of all admissible points.
Hence in the worst case the number of dimensions is  
the product of the number of couples (campaign,creative),
the number of locations and the number of time frames considered.

\begin{remark}
In practise we do not know the supply and cannot decide
in advance how many impressions of a given creative should
be delivered.
Thus we need to translate the number of impressions $x_{i,j,k,l}$
in terms of probability of delivery.
We transform a configuration into a
map that associates a couple $(k,l)$ with the probability
of delivery of $B_{i,j}$ for 
all possible corresponding couples $(i,j)$.
\end{remark}

Given a set of $t$-uples $C$,
we introduce the following
notation for subsets of $t-1$-uples:

\begin{align*}
C[i\rightarrow\alpha] = 
\{&(x_1,\dots, x_{i-1}, x_{i+1}, \dots, x_t) \\
& |(x_1,\dots, x_{i-1}, \alpha, x_{i+1},\dots,x_t) \in C \}; \\
C[i\rightarrow*] = 
\{&(x_1,\dots, x_{i-1}, x_{i+1}, \dots, x_t)  \\
& |\exists v |
(x_1,\dots, x_{i-1}, v, x_{i+1},\dots,x_t) \in C \}.
\end{align*}
\noindent i.e.\ we are considering 
respectively 
\begin{itemize}
\item $(t-1)$-uples obtained from 
$t$-uples in $C$, where $i$-th component is $\alpha$, 
in which the $i$-th component has been removed, and 
\item $(t-1)$-uples obtained from 
$t$-uples in $C$ where the $i$-th component has been
removed independently of its value.
\end{itemize}

In the same way, if more components are removed in parallel,
we introduce the notation:
$C[i_1\rightarrow \alpha_1,\dots,  i_n\rightarrow \alpha_n]$
for $(t-n)$-uples, where
$\alpha_j \in \mathbb{N}\cup\{*\}$ and
$i_j \in \mathbb{N}$ for $j \in \{1,\dots, n \}$.

\smallskip

Pictorially,
we could see a single configuration
$C=(x_{i,j,k,l})_{i,j,k,l}$ as a tridimensional matrix:

\bigskip
 
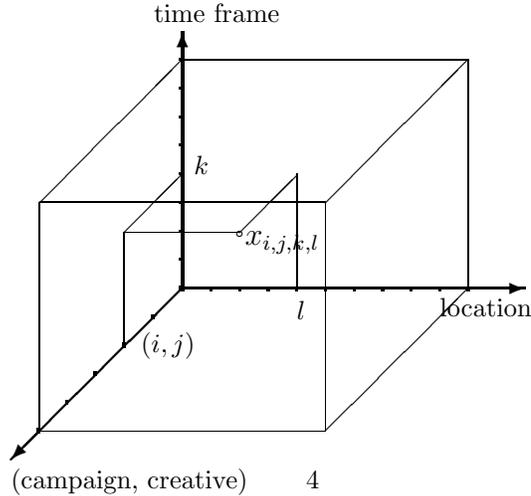
\setlength{\unitlength}{0.38mm}
    \begin{picture}(130,100)(-100,-10)
      \thicklines
      \put(0,0){\vector(1,0){120}}
      \multiput(0,0)(10,0){11}{\line(0,-1){1}}
      \put(90,-10){location}

      \put(0,0){\vector(0,1){90}}
      \multiput(0,0)(0,10){9}{\line(-1,0){1}}
      \put(-10,93){time frame}
      
      \linethickness{1.6pt}
      \put(0,0){\vector(-1,-1){60}}
      \multiput(0,0)(-10,-10){6}{\line(-1,0){1}}
      \put(-60,-70){(campaign, creative)}
      
      \thinlines
      \linethickness{0.2pt}
      \put(-14.5,-22.5){$(i,j)$}
      \put(-20.5,-20.5){\line(0,1){40}}
      \put(-20.5,19.5){\line(1,1){21}}
      \put(4,40){$k$}
      \put(-20.5,19.5){\line(1,0){40}}
      \put(22,15){{\large $x_{i,j,k,l}$}}      
      \put(20,19){\circle{1.39}}
      \put(20,19.5){\line(1,1){20}}
      \put(40,40){\line(0,-1){41}}
      \put(40,-11){$l$}

      \put(-50,-50){\line(0,1){80}}
      \put(-50,-50){\line(1,0){100}}
      \put(-50,30){\line(1,1){50}}
      \put(-50,30){\line(1,0){100}}

      \put(50,30){\line(1,1){50}}
      \put(50,-50){\line(1,1){50}}
      
      \put(50,-50){\line(0,1){80}}

      \put(0,80){\line(1,0){100}}
      \put(100,80){\line(0,-1){80}}      

    \end{picture}

\vspace{3.0cm}

\section{Realistic Model}

We want to consider a realistic model in which 
several constraints of different nature are taken into account.

\subsection{The Constraints}

We distinguish between the primary (physical)
constraints of the problem, the
secondary ones (commercial and optional) and 
the learning constraints (required by the learning
phase if it is included in the mathematical model).

\subsection{Primary constraints}
The primary constraints 
are given by the desired scheduling of the campaigns,
(the impossibility of certain creatives
at certain times and locations),
by the limited supply of impressions 
and by a (possibly) limited demand (campaign's budget):

\begin{enumerate}
\item\label{impossibleCon} 
The scheduling of the campaigns limits the admissible points:
certain creatives $B_{i,j}$ are only possible at 
certain time frames and locations. 
Typically a campaign (set of creatives) begins and ends at certain times
and its creatives are limited to certain locations, hours of the day,
days of the week, etc\dots
\item\label{supplyCon} Any location at a given time receives a limited supply 
of impressions (``location supply''), which solely depends on the traffic
of its page (more precisely on the portion of traffic given to the
ad-serving optimizer);
\item\label{demandCon} For any given campaign a 
given total profit may not be exceeded (``campaign demand'')
because only a finite campaign budget can be available.
\end{enumerate}

\begin{remark}
A campaign can have an unbounded budget, e.g.\ 
a campaign that only pays the client for an
actual purchase (highest step of a registration).
\end{remark}

\noindent For the sake of simplicity from now on we will assume implicitly 
that the first primary constraint (\ref{impossibleCon}) 
is always satisfied.

\subsection{Secondary constraints}
The secondary constraints may be of a commercial nature
and depend on the conditions in the contract between
the web-site's owner, the advertiser and possibly the
ad-serving company.
These constraints could be enforced, up to a certain extent, 
at real time while monitoring the delivery,
although having them as constraints is better for the 
optimality of the solution.

They are necessary to
increase the visibility of a certain campaign/creative:

\begin{enumerate}

\item\label{lastingCon} Any given creative/campaign 
should not last less than a given
period, e.g.\ the period in which the campaign is scheduled. 
We enforce this by setting a minimum for the number of 
impressions for each possible time frame.
\item\label{overflowCon} We would like to avoid having only 
one creative at a given location and time frame when more
than one choice is available.
\end{enumerate}

\subsection{Learning constraints}
Moreover, if the learning phase on the performance of new
creatives and new locations is to be included in the mathematical
model, additional constraints should be considered.
One way to force the system to learn on new creatives and
new nodes can be a constraint of this form:
for each new couple $(creative, node)$ we must have a minimum
number of impressions in all (or some initial) possible time frames.

\subsection{The Goal}
We want to maximize our expected revenue, which is
given by the total profit received
in a given configuration.

\section{A linear programming model}
Under the false but nevertheless mild
hypothesis that the impression profit
is constant with respect to its value      
we can assume that our
constraints are linear.
This assumption is not true in general
because 
there is no linear dependence between the total
profit generated by an impression-event and 
an impression value, i.e.\ 
displaying the same advertisement $x$ times 
on the same node,
possibly more than once to the same user,
does not necessarily produce $x$ times
the profit produced by one single display.

Since we are ultimately interested in 
the probability of delivery and since
integer linear programming is
computationally unfeasible (NP-hard),
a possible approach to this problem 
could be real linear programming:
we approximate our discrete problem
with a continuous one and we do not mind
considering a real number of impressions.

\subsection{Formalized constraints}
The points that do not contradict the first primary
constraint will form the unknowns
of our model.
In such a way we can avoid to include inequalities
for  the first primary constraint.

\subsubsection{Primary constraints}
We do not include the first primary constraints for the
reasons given above and assume that 
in our expressions all indices run over points
that do not contradict the first primary constraints.

Supply and demand are formalized as follows:

\noindent Second primary constraint:
\begin{equation}
\forall_{l,k} \; \sum_{i,j} x_{i,j,k,l} \leq 
\mathcal{S}_{l,k}; \; \; \text{(Supply)}\label{supplyIneq} 
\end{equation}
\noindent where
$\mathcal{S}_{l,k}$ is the supply at location $L_l$
and at time $T_k$.

\noindent Third primary constraint:
\begin{equation}
\forall_{i} \; \sum_{j,l,k} p_{i,j,k,l} x_{i,j,k,l} \leq 
\mathcal{D}_{i}; \; \; \text{(Demand)}\label{demandIneq}
\end{equation}
\noindent where $\mathcal{D}_{i}$ 
is the budget of the $i$-th campaign.

\subsubsection{Secondary constraints}
The secondary constraints are formalized as follows:

\noindent First secondary constraint:
\begin{equation}
\forall_{i,k} \; \sum_{j,l} x_{i,j,k,l} \geq 
\mathcal{\mu}_{i,k}; \; \; \text{(lasting)} \label{lastingIneq}
\end{equation}
\noindent where
$\mathcal{\mu}_{i,k}$ is the desired minimum 
delivery of
impressions of the $i$-th campaign at time $T_k$.

\noindent Second secondary constraint:
\begin{equation}
\forall_{l,k\in \mathcal{D}} \forall_{i,j} \;  x_{i,j,k,l} \leq P_{l,k} \cdot 
\mathcal{S}_{l,k}; \; \; \text{(no overflow)}\label{overflowIneq}
\end{equation}
\noindent 
where $P_{l,k} \in [0,1]$ (usually close to $1$)
defines how much a single creative can
occupy a location at a given time frame
and where $\mathcal{D}$ is the set of indices 
corresponding to locations and time frames
where at least $2$ different creatives are
possible.
\begin{remark}
The second secondary constraints (\ref{overflowIneq})
should only be limited to those cases in which
at a given location and time frame more than one
pair of campaign and creative is possible because
otherwise the constraint in (\ref{overflowIneq})
would prevent the node from being filled with impressions
even when this could be possible.
\end{remark}

\subsubsection{Learning constraints}
If the learning phase is included in the model 
a constraint should force the new creatives
and new locations to have a minimum number of impressions:
\begin{equation}
\forall_{\text{new } j,l} \, \forall_{i,k} \, x_{i,j,k,l} 
\geq \lambda_{i,j,k,l}.\label{learnIneq}
\end{equation}

We are also implicitly assuming that the unknowns are non-negative, 
i.e.\ 
\begin{equation}\label{positiveCon}
\forall_{i,j,k,l} \; x_{i,j,k,l} \geq 0.
\end{equation}

\section{Existence of a solution}
We see that there is no guarantee of consistency 
once the secondary and learning constraints are introduced,
even if we exclude the first primary constraints.
In general we need to solve the system of inequalities 
in order to make sure that there is indeed a solution.
Nevertheless if we carefully choose $\lambda_{i,j,k,l}$ and
$\mu_{i,k}$ we can avoid that the first secondary or learning
constraints contradict the second and third primary constraints.

In particular we have the following facts

\begin{fact}\label{lastingFact}
If we choose
\begin{equation}
\begin{split}
\mu_{i,k} \leq \min \{
& \frac{D_i}{|C[1\rightarrow i,2\rightarrow *,4\rightarrow *]| M},  \\
& \min_{t \in C[1\rightarrow i, 2\rightarrow *, 3\rightarrow k]}
\frac{S_{k,t}}{|C[2\rightarrow *, 3\rightarrow k, 4\rightarrow t]|}
  \} 
\end{split}
\end{equation}
\noindent then 
the semi-algebraic set defined by the inequalities
 (\ref{supplyIneq}), (\ref{demandIneq}), (\ref{lastingIneq})
and (\ref{positiveCon}) is not empty.
\end{fact}
\begin{proof}
The second secondary constraint (\ref{lastingIneq}) and our choice allow to have
\begin{equation}
\sum_{j,l} x_{i,j,k,l} 
\leq \frac{D_i}{|C[1\rightarrow i,2\rightarrow *,4\rightarrow *]| M}.
\end{equation}
Therefore for any campaign $C_i$ we have
\begin{equation}
\begin{split}
& \sum_{j,k,l} p_{i,j,k,l} x_{i,j,k,l} \leq \\
& \leq M \sum_{j,k,l}  x_{i,j,k,l} \leq M \sum_{k \in C[1\rightarrow i,2\rightarrow *,4\rightarrow *]} \sum_{j,l} x_{i,j,k,l} \leq \\
& \leq M \sum_{k \in C[1\rightarrow i,2\rightarrow *,4\rightarrow *]} 
\frac{D_i}{|C[1\rightarrow i, 2\rightarrow *, 4\rightarrow *]| M} = D_i.
\end{split}
\end{equation}
Moreover we can take
\begin{equation}
\sum_{j,l} x_{i,j,k,l} \leq  
\min_{t \in C[1\rightarrow i, 2\rightarrow *, 3\rightarrow k]}
\frac{S_{k,t}}{|C[2\rightarrow *, 3\rightarrow k, 4\rightarrow t]|}.
\end{equation}
Hence for any couple $(k,l)$ of location and time frame
we have
\begin{equation}
\begin{split}
\sum_{i,j} x_{i,j,k,l} & \leq \sum_s \sum_{i,j} x_{i,j,k,s} \leq 
\sum_i \sum_{k,s} x_{i,j,k,s} \leq \\
& \leq \sum_i \min_{t \in C[1\rightarrow i, 2\rightarrow *, 3\rightarrow k]}
\frac{S_{k,t}}{|C[1\rightarrow i, 2\rightarrow *, 3\rightarrow t]|} \leq \\
& \leq \sum_i \frac{S_{k,l}}{|C[2\rightarrow *, 3\rightarrow k, 4\rightarrow k]|}
= S_{k,l}.
\end{split}
\end{equation}
\end{proof}

\begin{fact}\label{learningFact}
If we choose
\begin{equation}
\lambda_{i,j,k,l} \leq \min \{
  \frac{D_i}{|C[1\rightarrow i]|M}, 
  \frac{S_{k,l}}{|C[3\rightarrow k,4\rightarrow l]|}
  \}
\end{equation}
\noindent then 
the semi-algebraic set defined by the inequalities
(\ref{supplyIneq}), (\ref{demandIneq}), (\ref{overflowIneq})
(\ref{positiveCon}) is not empty.
\end{fact}
\begin{proof}
The learning constraint (\ref{overflowIneq}) 
and our choice allow to have
\begin{equation}
x_{i,j,k,l} \leq  \frac{D_i}{|C[1\rightarrow i]|M}.
\end{equation}
\noindent from which it follows that for any campaign $C_i$:
\begin{equation}
\begin{split}
\sum_{j,k,l} p_{i,j,k,l} x_{i,j,k,l}  &\leq M \sum_{j,k,l} x_{i,j,k,l} \leq \\
& \leq M \sum_{j,k,l} \frac{D_i}{|C[1\rightarrow i]|M} = D_i.
\end{split}
\end{equation}
Moreover we can also take
\begin{equation}
x_{i,j,k,l} \leq \frac{S_{k,l}}{|C[3\rightarrow k,4\rightarrow l]|}.
\end{equation}
\noindent from which it follows that for any
couple $(l,k)$ of location and time we have
\begin{equation}
\sum_{i,j} x_{i,j,k,l} \leq 
\sum_{i,j} \frac{S_{k,l}}{|C[3\rightarrow k,4\rightarrow l]|} = S_{k,l}.
\end{equation}
\end{proof}

\begin{theorem}
The system of inequalities formed by 
second and third primary constraints,
first secondary constraints 
and learning constraints has at least a solution.
\end{theorem}
\begin{proof}
This follows from Fact~\ref{lastingFact} and
Fact~\ref{learningFact}.
\end{proof}

\begin{remark}
A slightly simpler model
in which only primary constraints (including the first ones)
are considered is an instance of the 
``Hitchcock's style transportation problem'' \cite{hitchcock}.
\end{remark}

\subsection{The objective function}
We want to maximize our expected revenue, which is
given by the sum of all expected profits received
in a given configuration.

Hence we can estimate the
revenue by taking the weighted sum of the
expected revenue in a given configuration $C$.
Such a sum will be our ``objective function'':
\begin{equation}
F(c) = \sum_{i,j,k,l} p_{i,j,k,l} \; x_{i,j,k,l}.
\end{equation}
\noindent where $p_{i,j,k,l} \; x_{i,j,k,l}$ is
the expected profit generated by $B_{i,j}$
at location $L_l$ and at time $T_k$.

\section{Simplified models: transportation problems}
The standard algorithm for solving linear programming
problems is the well known simplex algorithm.
The simplex algorithm, although fast in many practical
situations, is exponential in the worst case
and does not scale well enough for the size of
problems we want to consider.
Moreover no general linear programming algorithm
is known to be strongly polynomial.

An important speed-up can be achieved by simplifying the
model we have considered.
In particular we could consider a simpler model that
could fall into the category of ``transportation problems''.
For such problems very efficient algorithms are known
such as the ``stepping stone algorithm'' 
(see \cite{dantzig} and \cite{transshipment})

The classical transportation problem is a linear programming
problem whose constraints describe demands $d_i$ 
to be met and supplies $s_j$ to  be delivered:
\begin{equation}
\forall_j \; \sum_i x_{i,j} = s_j; \; \; \;
\forall_i \sum_j x_{i,j} = d_i.
\end{equation}

We say that a transportation problem is
``balanced'' if the total demand equals the total supply.
Taking into consideration only balanced problems is not
a real restriction because one can always put oneself
into this case by adding an extra dummy supply or
extra dummy demand with zero cost/gain.

In our case demands could describe a required number
of creatives and supplies a the number of
impressions that can be shown in given locations
and time frames.

A more general transportation problem 
goes under the name of ``Hitchcock's style transportation
problem'' (\cite{hitchcock}) that better approximates
our problem, in that the demand is obtained by multiplying
the impression values by a factor (representing their profit). 
The constraints are of the following kind:
\begin{eqnarray*}
\forall_j \; \sum_i x_{i,j} & = s_j; \;  \\
\forall_i \; \sum_j c_{i,j} x_{i,j} & = d_i 
\end{eqnarray*}
\noindent where $c_{i,j}$ is the value (cost or profit)
associated to $x_{i,j}$.

\section{Maximizing the value of impressions}
A related problem to the one of
maximizing the revenue is that
the problem of minimizing the 
number of used impressions under
some constraints on the revenue
generated by each campaign.
This is equivalent to the problem of maximizing
the profit generated by a single impression.

In other words we want to maximize
the revenue and secondly minimize
the number of impressions such that
the maximum revenue is achieved, i.e.\
we want to maximize the average profit of impressions,
provided that the maximum revenue is achieved.

If we exclude secondary and learning constraints, 
this can be formalized as
a special instance of the ``Hitchcock's style transportation problem''
(see \cite{hitchcock}).

Given the constraints
\begin{eqnarray}
\forall_{k,l} & \sum_{i,j}x_{i,j,k,l}\leq s_{k,l} \\
\forall_{i}   & \sum_{j,k,l} p_{j,k,l} x_{i,j,k,l} = d_i 
\end{eqnarray}
\noindent we want to minimize the following
objective function
\begin{equation}
\sum_{i,j,k,l} x_{i,j,k,l}
\end{equation}
\noindent i.e.\ the total number of impressions.

\section{Projections}
In order to apply our optimization algorithms we need
to have at least a projection of the supply and 
a projection
of the expected profit of all impressions allowed by
the first primary constraint.
The supply and impression profits could be estimated by taking 
a proper weighted average from the historical data.
The projection should take into account different factors:
episodic factors and possibly different periodicities 
(daily, weekly, yearly, etc\dots).

\subsection{Projecting the profit}\label{sse:projProfit}
The profit of an impression-event may depend on 
the periodicity of its campaign and of its location.
Since the eCPM tends to change slowly in time, it can
be predicted better than the supply.
If we want to determine the expected profit for
an impression-event $(B_{i,j},T_k,L_l)$ we can take
some average profit from our historical data
on ``similar'' events.
Our strategy is to use the most accurate and recent
available information in the historical data.

A simplified version of our algorithm can be described
by the following procedure (more subcases are considered 
by the actual algorithm):
\begin{enumerate}
\item Try to find enough impression-events  
of the form $(B_{i,j},T,L_l)$ where $T$ is 
a similar time (possibly same day of the week and same hour)
starting from the closest dates first.
\item Try to find enough impression-events
of the form $(B,T,L_l)$ and $(B_{i,j},T,L)$ 
and take a weighted average of the two averages,
where $L$ is any location, $B$ is any creative,
$T$ is a similar time.
\item Try to do the same as in the previous step
but with $(B,T,L_l$) and $(B_{i,k},T,L)$ where
$B_{i,k}$ is a different creative belonging to the same campaign.
\item Try to find enough impression-events 
(only on same the campaign, only the same node, etc\dots).
\end{enumerate}

\subsection{Projecting the supply}
A model for the projection of the supply should take into 
consideration the periodicity of the location, 
i.e. some sites are more often visited in particular periods 
of the year, day, hours, etc\dots
More periodicities may concur, e.g.\ a site may be visited
more often in a specific day of the week 
and at a specific hour of the day, and may also have
an episodic surge in the number of visitors for a short
period for some unpredictable event.
A mathematical model that could describe the concurrent
effects of different periodicities could be that of
superimposing waves, where each wave describe a different
factor, e.g.\ a weekly factor and a contingent factor.

\subsubsection{Weighted average}
In many practical cases it is enough to consider a weighted average
of the supply in the previous two weeks and at similar
hours in a similar fashion as to procedure~\ref{sse:projProfit} used for
the projection of the profit.

\subsubsection{Machine learning}
Regression analysis through
machine-learning techniques such as support vector machines
can be a viable approach for the
problem of properly choosing the weights of the average
of the different ``features'' (e.g.\ periodicities).
Non-linear kernels could also be taken into considerations
if they perform significantly better for the data sets under
consideration.

\section{Dynamic and stochastic nature of our problem}
In reality this approach has a serious drawback:
we are making a \underline{very false} \underline{assumption}
because by applying linear programming 
we are assuming that in the period under consideration
the expected impression profits $p_{i,j,k,l}$,
the expected supplies $S_{l,k}$,
and our constraints do not change.
We are also erroneously assuming that $p_{i,j,k,l}$ is a constant
with respect to $x_{i,j,k,l}$.

\subsection{Non-linear and dynamic problem}
Our problem is in fact non-linear and dynamic. 
To make things worse its state depends on external factors that cannot 
always be forecast (e.g.\ new campaigns can come into play).
Hence we are forced to continuously readapt to the
new constraints. Thus we can expect a good performance 
if our expected impression profit and constraints
do not change too much in the period of time
under consideration.

\subsection{Learning phase}
An additional problem comes from the fact that
the system has to learn how new creatives and
new locations perform.
The corresponding theoretical problem goes under the name
of ``exploration-exploitation trade-off'', i.e.\ 
the challenge of deciding between learning how
some resources perform
versus exploiting the ones that have so far performed better.
In practice we need to decide between
displaying advertisements that
help the system learn about their performance
versus displaying
those that generate a more immediate reward.
This problem has been addressed in \cite{abenakamura}
where a technique based on the Gittins index has been used.

\subsection{How far into the future}
This also poses the problem of deciding how fast
we want to update our information and how
far in time we want our 
optimization to ``see'' our problem
(i.e. how globally we want to solve the
problem).
A global solution could be a very bad
one if the conditions of the problem
were to change too quickly.
We might want to give a different weight to
an expected profit far in time in order
to compensate for possible changes
and so limit the risk.

We must also take into consideration
the stochastic nature of our data.
For example we might use historical data to extract
the standard deviation for
the profit of the impression and use it
to better assess the risk.

A possibility could be to
have an adaptive or semi-automatic system 
in which
the time span given to the optimizer
is adaptively/manually adjusted when there is
a high probability of a significant change
in the constraints, e.g.\ a new campaign
is likely going to come into play, 
the estimation of profit of an impression is
not stable enough, etc\dots.

\section{Targeting users}
The approach we have so far presented optimizes the
delivery of advertisements in both space (nodes)
and time (time frames). It does not explicitly
take users' profiles into consideration.

Nevertheless the very same algorithms and code can
be used to take into account users' profiles
by encapsulating the profile information into 
the node information.
Therefore we should simply store a pair
$(node, profile)$ into a single ``extended node''.
The result would be that 
\begin{itemize}
\item the supply are projected
onto triples $(node, profile, time)$;
\item the eCPM and the delivery are computed for quintuples

$(campaign, creative, node, profile, time)$.
\end{itemize}

\noindent No modification of the code is necessary.

\section{Computational considerations}

The large number of unknowns and constraints in this general approach 
can pose a serious problem to its computable feasibility.

\subsection{Reducing dimensions and constraints}
We could reduce the number of dimensions
by clustering similar attributes, 
i.e.\ combinations of locations and time frames
(see \cite{abenakamura} for an approach to this problem)
or by simplifying our model. For example we could 
simplify our model as follows:
\begin{itemize}
\item We could restrict our optimization problem
to periods of time in which the time constraints
do not change. This greatly reduces the number of unknowns
but could also brings us to suboptimal solutions.
\item We could avoid considering the secondary and learning
constraints within the model and have them enforced
during the delivery.
\item We can use a time horizon, beyond which all 
the time frames are considered jointly.
\item We could use similar attributes as one single attribute,
e.g.\ an impression at location $\lambda_1$ at time
$t_1$ could perform similarly to a location $\lambda_2$
at some other time $t_2$.
\item Projecting the ``impression profits'' is a costly 
operation because of the sheer number of points to be considered.
This operation can be sped up
by assuming that similar points produce the same profit.
\end{itemize}

\subsection{Other optimization algorithms}
The simplex algorithm is not the only known efficient
algorithm for linear programming.
Interior-point algorithms provide a valid alternative
and are polynomial 
(more specifically ``weakly polynomial''\footnote{
An algorithm is ``strongly polynomial'' if and only if
\begin{itemize}
\item the number of operations in the arithmetic model of computation 
is bounded by a polynomial in the number of integers in the 
input instance; and
\item the space used by the algorithm is bounded by a polynomial in 
the size of the input.
\end{itemize}
An algorithm which runs in polynomial time but which is not 
strongly polynomial is said to run in weakly polynomial time.
The existence of a strongly polynomial algorithm for real
linear programming is still an open problem.
}.
Totally different approaches to the optimization
could be possible, e.g.\ gradient based, genetic, etc\dots
Unfortunately these other approaches are suboptimal because
they are intrinsically local.
Moreover they do not exploit the linear nature of the constraints
and of the objective function.

\subsection{Tuning the supply projection}
Properly projecting the supply from historical data can be
a hard task due to the fact that the available historical data 
might not correspond to the real traffic but only to a 
possible variable portion of the real traffic which is given to the
ad-server.
This problem might be impossible if no regularity is present
in the data. A machine learning approach may better tackle
it than taking a simple weighted average of selected profits
at some previous periods with constant weights.
For instance this is the case if an optimal method is non-linear.

\section{Benchmarking linear solvers for ad-serving problems}
In order to test the computational feasibility of
two of the main linear programming solvers on our
problems we have used the free and open-source 
\texttt{lp\_solve} \cite{lpsolve} 
and \texttt{glpk} \cite{glpk} libraries.
They both come as \texttt{C} libraries; both
have wrapper interfaces in higher level languages
such as \texttt{Java}.
The result was a clear win for \texttt{glpk} at least
for the type of problems under our consideration.
Other alternative libraries are
\texttt{bpmpd} \cite{bpmpd}, 
\texttt{soplex} \cite{soplex}, \texttt{pcx} \cite{pcx}.

\section{Software Implementation}
We have implemented an ad-server optimizer in 
both \texttt{C} and in \texttt{Java} (see \cite{adproblem}).
We have used the \texttt{glpk} library for 
solving the mathematical model and \texttt{libsvm}
(see \cite{libsvm}) for automatically learning 
how to project future supply (Internet traffic).

\subsection{Features}
Our implementation has the following features:
\begin{itemize}
\item{Constraints: } We consider in our model 
the three types of primary constraints 
and provide as option the secondary and learning constraints. 
\item{Projection of the profit: } 
We perform a projection of the impression profits
from historical data, by taking into account different
periodicities (e.g.\ hourly, daily). 
\item{Projection of the supply: } We perform a projection of the supply
from historical data, 
by both machine learning (support vector machine \texttt{libsvm} 
and/or 
by a fixed weighted average which both take into account 
different periodicities.
\item{Time horizon: } We can set a time horizon in order to reduce
the number of unknowns. 
\end{itemize}

\subsection{Extracting the data from the database}
Our implementation requires as input:
\begin{enumerate}
\item historical data necessary for projecting the
impression profits and the future supply,
\item campaign data (budgets for each campaign),
\item scheduling data (set of possible impression not contradicting the
first primary constraints).
\end{enumerate}

\subsection{Main steps of the optimization}
The algorithm can be roughly subdivided in the
following macro steps:
\begin{enumerate}
\item Historical data is read.
\item The past supply is extracted from the historical data.
\item The future impression profits are projected from the
historical data.
\item The future supply is projected from the past supply.
\item The mathematical model for the optimization problem is
constructed.
\item The model is solved.
\item The solution of the problem is translated in terms
of probabilities of delivery.
\item The delivery probability is only used on the very next
time frame.
\item This procedure is repeated on the next time frame.
\end{enumerate}

\section{Results on real data provided by Neodata}
Our prototype has been used on some real data used at Neodata
and has been compared against the results produced by the
optimizer currently used at Neodata, which uses a simple
greedy algorithm: 
\begin{enumerate}
\item after a learning phase;
\item if a campaign is achieving its
target at the current rate, nothing is done,
otherwise, the campaign is stopped in its
less profit-generating nodes.
\end{enumerate}

The data we used were the logs and schedules creatives
used by two clients of Neodata, which, we call $A$ and $B$.
We have considered the data of April 2010 for both
companies.
We must remark that the percentage of traffic that is
managed by Neodata, neither is the total traffic nor
is it a constant percentage of the traffic generated
by the sites under consideration.
This makes the problem of properly estimating the
supply much harder (or even impossible).

The prototype achieved the following results:
$A$ was optimized equally well by the current optimizer
and our prototype; whereas $B$ was optimized better by a large margin
(more then 20\%) by our code.
We do not know for sure why the data on $A$ 
are not optimized equally well.
Possible reasons are: there is no room for further improvement,
the data on the supply cannot be used for the projection
because it does no correspond to a constant percentage 
of the real traffic.

The data was used as follows: the initial portion of the month
(e.g.\ the first 20 days) were used for training the system,
i.e.\ projecting the supply (traffic) and the profits.
The remaining part of the month was used as a schedule and
was optimized.

\section{Conclusion}
Our prototype has shown that real data can be indeed
optimized better than what a greedy algorithm
does.
There are still some open issues: how
to correctly project the supply when the 
conditions of the problem change quickly and
the data does not correspond to a constant
percentage of the traffic.

\bibliographystyle{alpha}
\bibliography{ImpOptimizer}

\newcommand{\etalchar}[1]{$^{#1}$}

\end{document}